\documentclass{svproc}
\usepackage{graphicx}
\usepackage{multicol}
\usepackage{footmisc}
\usepackage{amsmath}
\usepackage{amssymb}
\usepackage{subcaption}
\usepackage{bbm}

\usepackage{url}

\newtheorem{assumption}{Assumption}

\begin{document}

\mainmatter              % start of a contribution
\title{Value of Communication: Data-Driven Topology Optimization for Distributed Linear Cyber-Physical Systems}
\titlerunning{Data-Driven Communication Topology Optimization}  % abbreviated title (for running head)
\author{Michael Nestor\inst{1}\thanks{Corresponding author} \and Fei Teng\inst{1}}
%
%\authorrunning{Ivar Ekeland et al.} % abbreviated author list (for running head)
\tocauthor{Michael Nestor, Fei Teng}
\institute{\textsuperscript{1} Department of Electrical and Electronic Engineering, Imperial College London, UK\\
\email{m.nestor22@imperial.ac.uk}%,\\ WWW home page:
%\texttt{http://users/\homedir iekeland/web/welcome.html}
}

\maketitle              % typeset the title of the contribution

\begin{abstract}
Communication topology is a crucial part of a distributed control implementation for cyber-physical systems, yet is typically treated as a constraint within control design problems rather than a design variable. We propose a data-driven method for designing an optimal topology for the purpose of distributed control when a system model is unavailable or unaffordable, via a mixed-integer second-order conic program. The approach demonstrates improved control performance over random topologies in simulations and efficiently drops links which have a small effect on predictor accuracy, which we show correlates well with closed-loop control cost.
%The abstract should summarize the contents of the paper
%using at least 70 and at most 150 words. It will be set in 9-point
%font size and be inset 1.0 cm from the right and left margins.
%There will be two blank lines before and after the Abstract. \dots
% We would like to encourage you to list your keywords within
% the abstract section using the \keywords{...} command.
\end{abstract}

\section{Introduction}
%

%\begin{eqnarray*}
%  \dot{x}&=&JH' (t,x)\\
%  x(0) &=& x(T)
%\end{eqnarray*}

%
\subsection{Context and Motivation}
The communication topology and its impact on control performance is an emerging topic within the field of distributed control. Distinguished from centralized control, where all local agents must communicate with the central control agent, and decentralized control, where there is no communication at all, distributed control of interconnected dynamical systems has the flexibility to design which agents each agent should send information to or receive information from, even during run time given the increasing reconfigurability of communication solutions, such as software-defined networks. A higher degree of connectivity in the communication graph will improve control performance as each agent has a better picture of the global system, but will lead to increased communication overheads. This paper focuses on how to balance these competing objectives of minimizing communication whilst achieving satisfactory control performance. \par

Distributed control that achieves cooperation through communication has been proposed as offering better performance than decentralized schemes whilst providing improved flexibility, privacy and computational efficiency compared to centralized control \cite{Daoutidis-et-al-2018-Decomposing-Complex-Plants}. Consider the example of power distribution grids; significant growth of distributed energy resources such as electric vehicles will require coordination to provide grid services and avoid violating grid constraints. A distributed control structure mirrors the distributed nature of these resources, with the potential for enhanced flexibility and cooperation between, e.g., virtual power plant operators or prosumer communities \cite{Han-2018-Taxonomy-Evaluation-Distributed-Control-DERs}. \par

There must be a defined scheme for information exchange between agents in order to implement distributed control. All-to-all communication may be optimal from a control perspective, but is likely to lead to prohibitive communication and computation overheads in large-scale, complex systems \cite{Jovanovic-Dhingra-2016-Controller-Architectures-Tradeoffs}. Instead defining agent \emph{neighborhoods} reduces the communication burden, where agents communicate with the agents in their neighborhood; however, determining these neighborhoods is not a trivial task. Poor choices may lead to inefficient communication at best and unsatisfactory control performance at worst. 
%If an accurate model of the global system is available, 
The neighborhoods can be determined from a system dynamics model, if a suitably accurate model is available. However, identifying such a model can be expensive and time consuming. Recently, direct data-driven control methods that avoid model identification have gained traction \cite{Coulson-2019-Shallows-of-the-DeePC} \cite{Berberich-2019-DDMPC-Stability-Robustness} \cite{De-Persis-2020-Formulas-Data-Driven-Control}. However, it is unclear, under the data-driven setting, how to design adequate and yet efficient communication topology.

\vspace{-0.19mm}
\subsection{Literature Review}

Typical distributed control problems take the communication topology to be a constraint rather than a design variable; the System Level Synthesis framework \cite{Anderson-et-al-2019-SLS} treats controller sparsity structures as a constraint within the distributed control design problem. Distributed model predictive control (MPC) algorithms such as \cite{Conte-2016-Distributed-Synthesis-Stability-Cooperative-Distributed-MPC-Linear} and data-driven distributed MPC schemes \cite{Kohler-2022-Data-Driven-Distributed-MPC-Coupled-Linear} \cite{Alonso-2022-Data-Driven-Distributed-Localized-MPC} assume a known a-priori set of neighbors for each agent.
%, where the set is either known a-priori or defined by a non-zero state space interconnection matrix between subsystems. 
%Recent data-driven distributed predictive control algorithms for linear time-invariant (LTI) systems assume that the neighbors of each agent are pre-defined \cite{Kohler-2022-Data-Driven-Distributed-MPC-Coupled-Linear} or are fixed by the model-based criterion that the interconnection matrix is non-zero \cite{Alonso-2022-Data-Driven-Distributed-Localized-MPC}. 
State feedback control in \cite{Celi-2023-Distributed-Data-Driven-Control-Network-Sys} guarantees convergence for an arbitrary communication graph but no information is provided concerning the influence of the topology on the convergence rate. \par

Previous works addressing communication topology design for control include sparsity-promoting optimal controllers \cite{Dorfler-et-al-2014-Sparsity-Promoting-Optimal-WAC} and co-design of topology and distributed controller to balance communication costs against performance \cite{Gross-2011-Optimized-Distributed-Control-Network-Topology-Design}. 
%These approaches both solve a centralized design problem. An alternative approach is taken in \cite{Schuh-2016-Self-Organizing-Distributed-Control} where a controller is already known and the topology is adjusted online based on thresholding of the control signals. 
Graph-based methods include community detection \cite{Daoutidis-et-al-2018-Decomposing-Complex-Plants} and an optimization framework to find stabilizing control structures \cite{Mosalli-Babazadeh-2022-Stabilizing-Control-Structures-Optimization-Framework}. We refer to \cite{Chanfreut-et-al-2021-Survey-Clustering-Methods} for a detailed survey. These works assume knowledge of the global system model, and provide a basis for this data-driven approach. Recent work \cite{Coraggio-2023-Data-Driven-Design-Complex-Network-Structures} uses data-driven methods for network design; this addresses the underlying physical network structure rather than the communication topology of a cyber-layer controller. 

%Previous works have investigated centralized co-design of the communication topology and the distributed controller to balance the communication costs against control performance \cite{Gross-2011-Optimized-Distributed-Control-Network-Topology-Design} including back-up plans in case of link failures \cite{Gross-2013-Design-Distributed-Controllers-Comms-Topologies-Link-Failures} and a decomposition method to improve scalability \cite{Jilg-2013-Optimized-Distributed-Control-Topology-Design-Hierarchically-Interconnected}. An alternative approach is taken in \cite{Schuh-2016-Self-Organizing-Distributed-Control} where a controller is already known and the topology is adjusted online based on thresholding of the control signals. These works assume knowledge of the global system model, and provide a basis for this data-driven approach. Recent work \cite{Coraggio-2023-Data-Driven-Design-Complex-Network-Structures} uses data-driven methods for network design; this addresses the fundamental structure of the underlying network rather than the communication structure of the cyber-layer controller. \emph{Add notes about sparsity-promoting control structures, community detection methods, event-triggered MPC, SLS, plug-and-play via switching topologies.}
\vspace{-0.2mm}
\subsection{Main Contributions and Paper Structure}
\vspace{-0.2mm}

We consider a linear time-invariant (LTI) networked system consisting of subsystems coupled through their outputs. The problem we address is to design a communication topology suitable for a distributed control implementation without prior knowledge of a system model. We accomplish this through a centralized optimization problem carried out by a system coordinator, where we collect persistently exciting data from the global system that can be used to directly synthesize a communication topology and a multi-step output predictor using a least-square prediction error method. We frame our work in the multi-agent setting and take a modular approach such that we only tackle the topology design via the coordinator to set up the control design problem for the individual control agents to analyse; we assume that the topology is (re-)designed on a much slower timescale than control decisions are taken on. We co-optimize predictor and topology within a mixed-integer optimization problem to find a topology that balances communication costs and prediction accuracy and is optimal relative to the specified communication link cost weights.
We validate the topology optimization through simulations, showing that prediction cost increases as links are dropped due to greater communication costs, and investigate the connection between the topology design and closed-loop control cost via distributed MPC simulations.

In summary, this paper makes the following original contributions:

- to derive a mixed-integer second-order conic optimization problem (MISOCP) that directly uses data to compute an optimal communication topology and linear predictor and to assess hyperparameter tuning choices.

%- to consider the impact of the communication topology on the accuracy of an observer designed using Canonical Correlation Analysis (CCA), and consider if knowledge of the observer should affect the topology design.

%- to prove that a controller based on the optimized predictor leads to a control cost that is a lower bound for a controller designed directly using data and the designed topology.

- to find an upper bound on the prediction error induced by a given topology and the associated error in an open-loop control cost

- to analyse the trade-off between the prediction error and the number of communication links (and the associated communication costs), and the nature of the connection between this trade-off and the realised control cost.

The mathematical formulation of the distributed dynamical system is given in Section \ref{section:formulation}. In Section \ref{section:optimization} we present the mixed-integer conic optimization problem and derive bounds on the prediction error dependent on the chosen topology before simulation results and conclusions are given in Section \ref{section:simulation} and Section \ref{section:conclusion}.

%\subsection{Paper Structure}

%The mathematical formulation of the distributed dynamical system we consider is given in Section \ref{section:formulation}. In Section \ref{section:optimization} we derive the mixed-integer conic optimization problem, discussion on the relation to model-based approaches and control design are presented in Section \ref{section:stability}, simulation results are given in Section \ref{section:simulation}, before conclusions are made in Section \ref{section:conclusion}.

%%%%%%%%%%%%%%%%%%%%%%%%%%%%%%%%%%%%%%%%%%%%%%%%%%%%%%%%%%%%%%%%%%%%%%%%%%%%%%%%
\section{Preliminaries}
\label{section:preliminaries}
We write a $T$-sample long signal composed of the sequence $\{x(k)\}_{k=1}^\top$ in vector form as $x_{[1:T]} = [x^\top(1) \; ... \; x^\top(T)]^\top$. Its Hankel matrix form is defined by: 
\begin{eqnarray*}
    H_N(x) = \begin{bmatrix}
        x(1) & x(2) & \cdots & x(T - N + 1) \\
        x(2) & x(3) & \cdots & x(T - N + 2) \\
        \vdots & \vdots & \ddots & \vdots \\
        x(N) & x(N+1) & \cdots & x(T)
    \end{bmatrix}
\end{eqnarray*}
\begin{definition}[Persistency of Excitation]
    %Persistency of Excitation - 
    a signal $x_{[1:T]}$ where $x(k) \in \mathbb{R}^n$ is persistently exciting of order $L$ if $\mathrm{rank}(H_L(x)) = nL$.
\end{definition}
%\textbf{Definition 1}: Persistency of Excitation - a signal $x_{[1:T]}$ where $x(k) \in \mathbb{R}^n$ is persistently of order $L$ if $\mathrm{rank}(H_L(x)) = nL$. \\
We denote the matrix (block-) diagonal with $n$ (block-) diagonal elements $A_i$ by $\mathrm{diag}\{A_i\}_{i=1}^n$, and use the notation that a matrix A:
\allowdisplaybreaks{
\begin{equation*}
    A = \begin{bmatrix}
        A_{11} & \cdots & A_{1 N_2} \\
        \vdots & \ddots & \vdots \\
        A_{N_1 1} & \cdots & A_{N_1 N_2}
    \end{bmatrix}
\end{equation*}
}
may be represented by $A = \{A_{ij}\}_{i=1:N_1, \; j=1:N_2}$, and if $N_1 = N_2$ we can write $A = \{A_{ij}\}_{i,j=1:N_1}$. $\mathbb{I}^{a \times b}$ is the identity matrix of dimension $a \times b$, $\mathbf{0}^{a \times b}$ is the zero matrix of dimension $a \times b$ and $\mathbbm{1}^{a \times b}$ is a matrix of ones of dimension $a \times b$. $\mathbb{Z}_{>0}$ is the set of positive integers.

%\vspace{-1mm}
\section{System Formulation}
\label{section:formulation}
%\vspace{-1mm}

We consider a set of subsystems $\mathcal{V} = \{1,...,M\}$, where $M$ is the total number of subsystems, within a multi-agent cyber-physical system (CPS) framework. Under this approach, we abstract two directed graphs; one for the physical layer, $\mathcal{G}_P = (\mathcal{V}, \, \mathcal{E}_P)$, and a second for the cyber layer, $\mathcal{G}_C = (\mathcal{V}, \, \mathcal{E}_C)$, with $\{ \mathcal{E}_P, \, \mathcal{E}_C \} \subseteq \mathcal{V} \times \mathcal{V}$. The physical layer graph represents the physical subsystems and the interconnections between them, whilst the cyber layer graph represents the control agents (one for each subsystem) and the cyber layer communication between the agents. Therefore, both graphs share the same vertex set but have differing edge sets.
%We define the graphs by:
%begin{subequations}
%\begin{eqnarray}
%    \mathcal{G}_P &=& (\mathcal{V}, \, \mathcal{E}_P) \\
%    \mathcal{G}_C &=& (\mathcal{V}, \, \mathcal{E}_C) \\
%    \{ \mathcal{E}_P, \, \mathcal{E}_C \} &\subseteq& \mathcal{V} \times \mathcal{V}.
%\end{eqnarray}
%\end{subequations}
The edge set $\mathcal{E}_P$ is defined by the physical interconnections between subsystems; i.e., $(j,i) \in \mathcal{E}_P$ if the physical state of subsystem $i \in \mathcal{V}$ is directly influenced by subsystem $j \in \mathcal{V}$. Meanwhile, $\mathcal{E}_C$ is defined by the communication links between agents; $(j,i) \in \mathcal{E}_C$ if agent $j \in \mathcal{V}$ sends information to agent $i \in \mathcal{V}$. We write $\mathcal{N}_i^C = \{ j : (j,i) \in \mathcal{E}_C \}$, with $\mathcal{N}_i^C$ denoting the cyber neighborhood of agent $i$ (the set of agents sending information to agent $i$). The Boolean variable $\delta_{ij} \in \{0, 1 \}$ represents whether a communication link exists from agent $j$ to agent $i$:
\begin{eqnarray}
\label{eqn:define_delta}
    \delta_{ij} = 
    \begin{cases}
        1 \ \textrm{if} \ (j,i) \in \mathcal{E}_C \\
        0 \ \textrm{otherwise}.
    \end{cases}
\end{eqnarray}
There is a cost associated with increasing the total number of communication links due to, e.g., bandwidth limitations, and phenomena such as rapid increases in packet loss when the communication channel utilization threshold is breached \cite{Wu-et-al-2020-Towards-Distributed-SDN}. We make a simplifying assumption that each link has a fixed cost.
\begin{assumption}
\label{assumption:comms_costs}
    There is a scalar link cost $c_{ij} \in \mathbb{R}_{\geq 0}$ associated with every possible cyber layer edge $(j,i) \in \mathcal{V} \times \mathcal{V}$.
\end{assumption}

A typical control problem is to determine the control inputs to minimize some cost as a function of the inputs and outputs of all the agents over a (possibly infinite) time horizon, subject to the system dynamics and without leaving a feasible region that inputs and/or outputs must stay within. In the multi-agent setting, each agent calculates its control input with a control function $\kappa_i(\cdot)$:
\begin{eqnarray}
    u_i(k) = \kappa_i(\xi_i(k), \theta_i, \xi_{\mathcal{N}_i^C}(k), \theta_{\mathcal{N}_i^C}),
\end{eqnarray}
where designing $\kappa_i(\cdot)$ is the control design problem, $\theta_i$ is a vector of local problem parameters such as the weights of the cost of $u_i(k)$ and $y_i(k)$, and $\xi_i(k)$ is an extended state vector of recent local inputs and output measurements (i.e., $\xi_i(k) = [u_i^\top(k-1) \; ... \; u_i^\top(k-N_P) \; y_i^\top(k-1) \; ... \; y_i^\top(k-N_P)]^\top$ for some $N_P \in \mathbb{Z}_{>0}$). Alongside local information $\theta_i$ and $\xi_i(k)$, agent $i$ also has information $\theta_{\mathcal{N}_i^C} = \{ \theta_j : j \in \mathcal{N}_i^C \}$ and $\xi_{\mathcal{N}_i^C} = \{ \xi_j(k) : j \in \mathcal{N}_i^C \}$ from its set of neighboring agents $\mathcal{N}_i^C$ that it can use to improve its decision making and reduce the control cost. We frame our approach with the following problem definition:
\begin{problem}
Design the cyber layer edge set $\mathcal{E}_C$ with knowledge of the set of subsystems and agents $\mathcal{V}$, but without knowing the physical layer edge set $\mathcal{E}_P$.
\end{problem}
%We will use data to design $\mathcal{E}_C$ by making implicit inferences about $\mathcal{E}_P$. The idea is that (stronger) physical interconnections between subsystems means that it is likely (greater) reductions in closed-loop control cost can be made by sharing information between their controlling agents, enabled by including an edge connecting these agents in the cyber layer communication graph.
%\par 
If two subsystems are physically interconnected, then enabling information sharing between the controlling agents of the subsystems will non-strictly decrease their closed-loop control cost. Stronger interconnections are likely to lead to a greater reduction in cost. Information sharing is enabled by an edge connecting these agents in the cyber layer communication graph.
%The control agents of physically interconnected subsystems are likely to reduce their closed-loop control costs by sharing information. Greater reductions are likely when the interconnections are stronger. Information sharing is enabled by an edge connecting these agents in the cyber layer communication graph.
%Control agents are likely to reduce their closed-loop control costs if they share information
\par

A distributed LTI system where subsystems are coupled through outputs and outputs are subject to noise may be described by:
\begin{subequations}
\label{eqn:subsys_dynamics}
\begin{eqnarray}
    \label{eqn:subsys_state_dyn}
    x_i(k+1) &=& A_i x_i(k) + B_i u_i(k) + \sum_{j=1}^M {E_{ij} y_j(k)} \\
    \label{eqn:subsys_output_eqn}
    y_i(k) &=& C_i x_i(k) + D_i u_i(k) + v_i(k),
\end{eqnarray}
\end{subequations}
where for subsystem $i \in \mathcal{V}$ and time step $k$, $x_i(k) \in \mathbb{R}^{n_i}$ is the state, $u_i(k) \in \mathbb{R}^{m_i}$ is the input, $y_i(k) \in \mathbb{R}^{p_i}$ is the output and $v_i(k) \in \mathbb{R}^{p_i}$ is zero-mean white noise. The interconnection matrix $E_{ij} = \mathbf{0}^{n_i \times p_j} \; \forall i=j$. The global system dynamics can be written as:
\begin{subequations}
\label{eqn:global_dynamics}
\begin{eqnarray}
    \label{eqn:global_state_dyn}
    x^{(g)}(k+1) &=& A x^{(g)}(k) + B u^{(g)}(k) \\
    \label{eqn:global_output_eqn}
    y^{(g)}(k) &=& C x^{(g)}(k) + D u^{(g)}(k) + v^{(g)}(k),
\end{eqnarray}
\end{subequations}
where the global state, input, output and noise vectors are given by $x^{(g)}(k) = [x_1^\top(k) \; ... \; x_M^\top(k)]^\top$, $u^{(g)}(k) = [u_1^\top(k) \; ... \; u_M^\top(k)]^\top$, $y^{(g)}(k) = [y_1^\top(k) \; ... \; y_M^\top(k)]^\top$ and $v^{(g)}(k) = [v_1^\top(k) \; ... \; v_M^\top(k)]^\top$, respectively. The global state-space matrices in \eqref{eqn:global_dynamics} are given by:
{\allowdisplaybreaks
\begin{subequations}
\begin{eqnarray}
    A &=& \begin{bmatrix} A_{1} & E_{12}C_2 & ... & E_{1M}C_M \\ E_{21}C_1 & A_{2} & ... & E_{2M}C_M \\ \vdots & \vdots & \ddots & \vdots \\ E_{M1}C_1 & E_{M2}C_2 & ... & A_{M} \end{bmatrix}, \
    B = \begin{bmatrix} B_{1} & E_{12}D_2 & ... & E_{1M}D_M \\ E_{21}D_1 & B_{2} & ... & E_{2M}D_M \\ \vdots & \vdots & \ddots & \vdots \\ E_{M1}D_1 & E_{M2}D_2 & ... & B_{M} \end{bmatrix} \  \\
    C &=& \mathrm{diag}\{C_i\}_{i=1}^M, \  
    D = \mathrm{diag}\{D_i\}_{i=1}^M, \ 
   %F = \begin{bmatrix} F_{1} & E_{12} & ... & E_{1M} \\ E_{21} & F_{2} & ... & E_{2M} \\ \vdots & \vdots & \ddots & \vdots \\ E_{M1} & E_{M2} & ... & F_{M} \end{bmatrix} 
\end{eqnarray}
\end{subequations}
}

%To enable linear systems analysis, we make the following assumption:
\begin{assumption}
\label{assumption:LTI_sys}
    The subsystems in the set $\mathcal{V}$ all follow dynamics \eqref{eqn:subsys_dynamics}, and the global physical system made up of the graph $\mathcal{G}_P$ obeys dynamics \eqref{eqn:global_dynamics}.
    %The set of subsystems obey linear time-invariant (LTI) dynamics and the subsystems are coupled through their outputs.
    The pair $(A,C)$ is observable, and the pair $(A,B)$ is controllable.
    Furthermore, whilst we consider systems obeying the dynamics in \eqref{eqn:global_dynamics}, we assume that the matrices $A_i$, $B_i$, $C_i$ and $D_i$ in \eqref{eqn:subsys_dynamics} are unknown $\forall i \in \mathcal{V}$, and $E_{ij}$ is unknown $\forall j \in \mathcal{V}, \ j \neq i$. 
    %($E_{ii} = \mathbf{0}^{n_i \times p_i}$).
\end{assumption}
%Therefore, we can describe the dynamics of subsystem $i$ by: \par

%\begin{assumption}
%\label{assumption:obsv_ctrb}
%    The pair $(A,C)$ is observable, and the pair $(A,B)$ is controllable.
%\end{assumption}
%\begin{assumption}
%\label{assumption:unknown_dynamics}
%    Although we consider systems obeying the dynamics in \eqref{eqn:global_dynamics}, we assume that the state-space matrices $A_i$, $B_i$, $C_i$ and $D_i$ in \eqref{eqn:subsys_state_dyn} and \eqref{eqn:subsys_output_eqn} are unknown $\forall i \in \mathcal{V}$, and $E_{ij}$ is unknown $\forall j \in \mathcal{V}, \ j \neq i$ ($E_{ii} = \mathbf{0}^{n_i \times p_i}$).
%\end{assumption}
Due to Assumption \ref{assumption:LTI_sys}, we use a central coordinator to generate an input signal $u^{(g),d}_{[1:T]}$ that obeys Assumption \ref{assumption:fundamental_lemma} to excite the dynamics; resulting output measurements allow an understanding of system behavior. The superscript $d$ denotes collected data used in topology design.
\begin{assumption}[\cite{Coulson-2019-Shallows-of-the-DeePC} Lemmas 4.1 and 4.2]
\label{assumption:fundamental_lemma}
    For $\{ T_{\textrm{ini}}, \, N \} \in \mathbb{Z}_{>0}$, the signal $u^{(g),d}_{[1:T]} \in \mathbb{R}^{m^{(g)}T}$ is persistently exciting of order $T_{\textrm{ini}} + N + n^{(g)}$; $m^{(g)} = \sum_{i=1}^M{m_i}$ and $n^{(g)} = \sum_{i=1}^M{n_i}$. $T$ is chosen such that $T \geq (m^{(g)} + 1)(T_{\textrm{ini}} + N + n^{(g)}) - 1 = T_{\textrm{min}}$, and $T_{\textrm{ini}} \geq \ell$, where $\ell$ is the system lag (that is, the smallest integer such that the observability matrix $[C^\top \; (CA)^\top \; ... \; (CA^{\ell - 1})^\top]^\top$ has rank $n^{(g)}$).
\end{assumption}
The corresponding collected output signal is given by $y^{(g),d}_{[1:T]} \in \mathbb{R}^{p^{(g)}T}$, with $p^{(g)} = \sum_{i=1}^M{p_i}$. The collected input and output data may be written as Hankel matrices $U^{(g),d} = H_{T_{\textrm{ini}} + N}(u^{(g),d}_{[1:T]})$, $Y^{(g),d} = H_{T_{\textrm{ini}} + N}(y^{(g),d}_{[1:T]})$. We permute and partition these matrices to end up with four subsystem-ordered data matrices:
\begin{align}
\label{eqn:Hankel_matrices}
    U^P = \begin{bmatrix}
        U^P_1 \\ \vdots \\ U^P_M
    \end{bmatrix}, \qquad
    Y^P = \begin{bmatrix}
        Y^P_1 \\ \vdots \\ Y^P_M
    \end{bmatrix},
    & \qquad U^F = \begin{bmatrix}
        U^F_1 \\ \vdots \\ U^F_M
    \end{bmatrix}, \qquad 
    Y^F = \begin{bmatrix}
        Y^F_1 \\ \vdots \\ Y^F_M
    \end{bmatrix}
\end{align}
where $\forall i \in \mathcal{V}: \ U_i^P = H_{T_{\textrm{ini}}}(u_{i \, [1:T - N]}^d)$, $Y_i^P = H_{T_{\textrm{ini}}}(y_{i \,[1:T - N]}^d)$, \\
$U_i^F = H_{N}(u_{i \, [N+1:T]}^d)$, $Y_i^F = H_{N}(y_{i \,[N+1:T]}^d)$, and the Hankel matrices are dimensioned by $U^P \in \mathbb{R}^{T_{\textrm{ini}} m^{(g)} \times (T - (T_{\textrm{ini}} + N) + 1)}$, $Y^P \in \mathbb{R}^{T_{\textrm{ini}} p^{(g)} \times (T - (T_{\textrm{ini}} + N) + 1)}$, $U^F \in \mathbb{R}^{N m^{(g)} \times (T - (T_{\textrm{ini}} + N) + 1)}$ and $Y^F \in \mathbb{R}^{N p^{(g)} \times (T - (T_{\textrm{ini}} + N) + 1)}$. Note we have separated the first $T_{\textrm{ini}}$ from the remaining $N$ block rows of the Hankel matrices. If there is zero noise, the data matrix $\begin{bmatrix}
    U^{(g),d^\top} & Y^{(g),d ^\top}
\end{bmatrix}^\top$ will span the subspace of trajectories of length $T_{\textrm{ini}} + N$ if it has rank $N m^{(g)} + n^{(g)}$ \cite{Markovsky-Dorfler-2021-Behavioral-Systems-Theory-Analysis-Signal-Processing-Control}. This result forms the cornerstone of the behavioral control methods \cite{Coulson-2019-Shallows-of-the-DeePC} \cite{Berberich-2019-DDMPC-Stability-Robustness} and is aligned with the Subspace Predictive Control (SPC) approach \cite{Favoreel-et-al-1999-SPC}. The SPC methodology is based on subspace identification and makes use of a linear output predictor:
\begin{eqnarray}
\label{eqn:spc_predictor}
    \Hat{y}^{(g)} = K \begin{bmatrix}
        u^{(g)}_{\textrm{ini}} \\ y^{(g)}_{\textrm{ini}} \\ u^{(g)}_f
    \end{bmatrix},
\end{eqnarray}
where $u^{(g)}_{\textrm{ini}} \in \mathbb{R}^{m^{(g)} T_{\textrm{ini}}}$ and $y^{(g)}_{\textrm{ini}} \in \mathbb{R}^{p^{(g)} T_{\textrm{ini}}}$ are the initialization trajectories (most recent $T_{\textrm{ini}}$ inputs and outputs)
%where $u^{(g)}_{\textrm{ini}} \in \mathbb{R}^{m^{(g)} T_{\textrm{ini}}}$ is the input initialization trajectory (most recent $T_{\textrm{ini}}$ inputs), $y^{(g)}_{\textrm{ini}} \in \mathbb{R}^{p^{(g)} T_{\textrm{ini}}}$ is the output initialization trajectory (most recent $T_{\textrm{ini}}$ outputs), 
$u^{(g)}_{f} \in \mathbb{R}^{m^{(g)} N}$ is the future $N$-step input sequence and $y^{(g)}_{f} \in \mathbb{R}^{p^{(g)} N}$ is the $N$-step output prediction trajectory. $K \in \\ \mathbb{R}^{p^{(g)} N \times (m^{(g)} + p^{(g)})T_{\textrm{ini}} + m^{(g)} N}$ is the predictor matrix. In Section \ref{section:optimization} below, we will examine a standard least-norm optimization problem to determine $K$ before including the influence of the communication topology within the prediction. First, we make a final assumption building on Assumption \ref{assumption:LTI_sys} to ensure the existence of a stabilizing control law for the designed topology.
\begin{assumption}
\label{assumption:no_DFMs}
    The parameters of the global system dynamics defined by \eqref{eqn:global_dynamics} are such that the global system does not have any unstable decentralized fixed modes (DFMs) \cite{Davison-Ozguner-1983-Characterizations-DFMs}, hence there exists a decentralized control law (i.e., when $\delta_{ij} = 0 \ \forall \{ i, j \} \in \mathcal{V}$) that results in a stable closed-loop system response; i.e., each subsystem can stabilize itself despite the coupling with neighbors.
\end{assumption}
Note that each agent $i$ has access to the local information $u_i$ and $y_i$ required for a decentralized control implementation where $u_i(k) = \kappa_i(\theta_i, \xi_i(k))$ without communication, hence Assumption \ref{assumption:no_DFMs} holds if there are no unstable DFMs.
%Using Assumptions \ref{assumption:obsv_ctrb} and \ref{assumption:fundamental_lemma}, we state the following Lemma reproduced from \cite{Coulson-2019-Shallows-of-the-DeePC} that is a consequence of the Fundamental Lemma \cite{Willems-2005-Fundamental-Lemma} in behavioral systems theory and forms the basis for behavioral control methods such as \cite{Coulson-2019-Shallows-of-the-DeePC} \cite{Berberich-2019-DDMPC-Stability-Robustness} \cite{De-Persis-2020-Formulas-Data-Driven-Control}.
%\begin{lemma}[see \cite{Willems-2005-Fundamental-Lemma} and \cite{Coulson-2019-Shallows-of-the-DeePC}]
%    Assume that Assumptions \ref{assumption:obsv_ctrb} and \ref{assumption:fundamental_lemma} are met and neglect noise $v^{(g)}$. Then any valid $T_{\textrm{ini}} + N$-length trajectory of the system defined by \eqref{eqn:global_state_dyn} and \eqref{eqn:global_output_eqn} may be represented by:
%\begin{equation}
%\label{eqn:behavioral_prediction}
%    \begin{bmatrix}
%        U^{(g),d} \\ Y^{(g),d}
%    \end{bmatrix}
%    g = \begin{bmatrix}
%        \Bar{u}^{(g)} \\ \Bar{y}^{(g)},
%    \end{bmatrix}
%\end{equation}
%where $g \in \mathbb{R}^{T - (T_{\textrm{ini}} + N) + 1}$, $\Bar{u}^{(g)} = u^{(g)}_{[1:T_{\textrm{ini}} + N]}$ and $\Bar{y}^{(g)} = y^{(g)}_{[1:T_{\textrm{ini}} + N]}$.
%\end{lemma}

\section{Communication Topology and Predictor Optimization}
\label{section:optimization}

Traditional control approaches consist of a two-stage procedure of first estimating a parametric model of the system to be controlled before designing a controller to stabilize/steer the model according to the desired behavior, often with a third intermediate step to design a state estimator with the model. The SPC approach \cite{Favoreel-et-al-1999-SPC} instead identifies a predictor of future system behavior using least-squares regression, which can then be used in an explicit control law to determine the optimal input sequence.

\subsection{Least-Squares Predictor Identification}

A similar prediction error method (PEM) that is combined with MPC \cite{Huang-et-al-2019-DeePC-Grid-Connected-Power-Converters} finds an optimal linear predictor that is used in the MPC optimization problem; this is formulated as a bi-level optimization. The predictor is found by considering each column in the Hankel matrices \eqref{eqn:Hankel_matrices} as a trajectory; we can use \eqref{eqn:spc_predictor} to make a prediction $\Hat{Y}^F$ of $Y^F$; $\Hat{Y}^F = K \begin{bmatrix}
        U^P \\ Y^P \\ U^F
    \end{bmatrix}$.
The predictor $K$ optimizing the least-squares fitting criterion is found as the minimizer of \cite{Dorfler-2023-Bridging-Direct-Indirect-Data-Driven-Control}:
\allowdisplaybreaks{
\begin{subequations}
\label{eqn:SPC_LS}
\begin{eqnarray}
    &\min_{\Hat{K}} &\left\Vert Y^F - \Hat{K} \begin{bmatrix}
        U^P \\ Y^P \\ U^F
    \end{bmatrix} \right\Vert _F^2 \\
    &\textrm{s.t.}& \Hat{K} = [\Hat{K}_p^u \; \Hat{K}_p^y \; \Hat{K}_f] \\
    && \textrm{rank}([\Hat{K}_p^u \; \Hat{K}_p^y]) = n^{(g)} \\
    && \Hat{K}_f \; \textrm{Lower-Block Triangular},
\end{eqnarray}
\end{subequations}
}
where the rank constraint ensures LTI dynamics of order $n^{(g)}$ and the lower-block triangular constraint ensures causality in predictions. Supposing Assumptions \ref{assumption:LTI_sys} and \ref{assumption:fundamental_lemma} hold, then the Hankel data matrices \eqref{eqn:Hankel_matrices} in \eqref{eqn:SPC_LS} span the trajectory subspace \cite{Coulson-2019-Shallows-of-the-DeePC} and $K$ calculated with \eqref{eqn:SPC_LS} is able to accurately predict an output trajectory according to \eqref{eqn:spc_predictor}. When data is noise-free, an exact closed-form solution exists and is given by $Y^F$ multiplied by the pseudo-inverse of the concatenated remaining Hankel matrices (i.e., $K = Y^F \begin{bmatrix}
    U^{P^\top} & Y^{P^\top} & U^{F^\top}
\end{bmatrix}^{\top^\dagger}$) as the solution to the unconstrained least squares problem.

\subsection{Including the Communication Costs in Predictor Optimization}

We now consider the desire to minimize communication costs defined by Eq. \eqref{eqn:define_delta} and supposing Assumption \ref{assumption:comms_costs} holds; i.e., $\min_{\{ \delta_{ij} \}_{i,j=1:M}} \sum_{i,j=1}^{M} c_{ij} \delta_{ij}$ with the set of Boolean $\delta_{ij}$ variables defining $\mathcal{E}_C$. Consider that if $\delta_{ij} = 0$, there is no link transferring information from agent $j$ to agent $i$, therefore agent $i$'s predictions cannot be calculated using the measurements or decisions of agent $j$. By writing $K$ as:
\begin{eqnarray}
    K = \begin{bmatrix}
        K_{p,11}^u & \cdots & K_{p,1M}^u & K_{p,11}^y & \cdots & K_{p,1M}^y & K_{f,11} & \cdots & K_{f,1M} \\
        \vdots & \ddots & \vdots & \vdots & \ddots & \vdots & \vdots & \ddots & \vdots \\
        K_{p,M1}^u & \cdots & K_{p,MM}^u & K_{p,M1}^y & \cdots & K_{p,MM}^y & K_{f,M1} & \cdots & K_{f,MM}
    \end{bmatrix},
\end{eqnarray}
the communication topology influence on predictor structure is encoded by $\delta_{ij} = 0 \ \implies \ \begin{bmatrix}
        K_{p,ij}^u & K_{p,ij}^y & K_{f,ij}
    \end{bmatrix} = 0$.
%\begin{eqnarray}
%    \delta_{ij} = 0 \quad \implies \quad \begin{bmatrix}
%        K_{p,ij}^u & K_{p,ij}^y & K_{f,ij}
%    \end{bmatrix} = 0.
%\end{eqnarray}
We seek to choose the edge set $\mathcal{E}_C$ that minimizes communication costs and least-square prediction error by optimizing over $\delta_{ij} \; \forall \{ i,j \} \in \mathcal{V}$ and $\Hat{K}$ within a multi-criteria mixed-integer program:
\allowdisplaybreaks{
\begin{subequations}
\label{eqn:original_optimization}
\begin{eqnarray}
    \min_{\Delta, \Hat{K}}&& \sum_{i,j=1}^{M} c_{ij} \delta_{ij} + \left\Vert Y^F - \Hat{K} \begin{bmatrix}
        U^P \\ Y^P \\ U^F
    \end{bmatrix} \right\Vert _F^2 \\
    \textrm{s.t.}&& \delta_{ij} \in \{ 0,1 \} \quad \forall \{ i,j \} \in \mathcal{V} \\
    && \Delta = \{ \delta_{ij} \}_{i,j=1:M} \\
    && \Hat{K} = [\Hat{K}_p^u \; \Hat{K}_p^y \; \Hat{K}_f] \\
    \label{eqn:rank_constraint}
    && \textrm{rank}([\Hat{K}_p^u \; \Hat{K}_p^y]) = n^{(g)} \\
    \label{eqn:lbt_constraint}
    && \Hat{K}_f \; \textrm{Lower-Block Triangular} \\
    && \Hat{K}_p^u = \{ \Hat{K}_{p,ij}^u \}_{i,j=1:M}, \ \Hat{K}_p^y = \{ \Hat{K}_{p,ij}^y \}_{i,j=1:M}, \ \Hat{K}_f = \{ \Hat{K}_{f,ij} \}_{i,j=1:M} \  \\
    \label{eqn:implies_constraint}
    && \delta_{ij} = 0 \quad \implies \quad \begin{bmatrix}
        \Hat{K}_{p,ij}^u & \Hat{K}_{p,ij}^y & \Hat{K}_{f,ij}
    \end{bmatrix} = 0 \quad \forall \{i,j\} \in \mathcal{V}, \; i \neq j.
\end{eqnarray}
\end{subequations}
}
In order to solve \eqref{eqn:original_optimization}, we first perform a convex relaxation and drop the rank \eqref{eqn:rank_constraint} and causality \eqref{eqn:lbt_constraint} constraints. Our approach is to increase the data trajectory length $T$ and average over multiple datasets to reduce the effect of noise and implicitly the need for the rank constraint. Ensuring a causal predictor structure is beneficial but not a strict requirement for predictive control, see \cite{Dorfler-2023-Bridging-Direct-Indirect-Data-Driven-Control}; non-causal predictors can still achieve acceptable performance. Next, we perform a convex big-M relaxation \cite{Gross-2011-Optimized-Distributed-Control-Network-Topology-Design} on the implication constraint \eqref{eqn:implies_constraint}; that is, we replace it with:
\allowdisplaybreaks{
\begin{eqnarray}
\left. \begin{array}{c}
    -\overline{M} \delta_{ij} \mathbbm{1}^{p_i N \times m_j T_{\textrm{ini}}} \leq \Hat{K}_{p,ij}^u \leq \overline{M} \delta_{ij} \mathbbm{1}^{p_i N \times m_j T_{\textrm{ini}}} \\
    -\overline{M} \delta_{ij} \mathbbm{1}^{p_i N \times p_j T_{\textrm{ini}}} \leq \Hat{K}_{p,ij}^f \leq \overline{M} \delta_{ij} \mathbbm{1}^{p_i N \times p_j T_{\textrm{ini}}} \\
    -\overline{M} \delta_{ij} \mathbbm{1}^{p_i N \times m_j N} \leq \Hat{K}_{f,ij} \leq \overline{M} \delta_{ij} \mathbbm{1}^{p_i N \times m_j N}
\end{array} \right\} \; \forall \{i,j\} \in \mathcal{V}, \; i \neq j.
\end{eqnarray}}
This relaxation ensures \eqref{eqn:implies_constraint} holds both when $\delta_{ij} = 0$, and when $\delta_{ij} = 1$ assuming that all elements of $\Hat{K}$ are bounded by $-\overline{M} \leq \Hat{K}_{ij} \leq \overline{M}$. The choice of $\overline{M}$ should be as small as possible to improve solver efficiency but large enough that it does not constrain the predictor sub-matrices from giving accurate predictions. We arrive at a mixed-integer quadratic program:
\begin{subequations}
\label{eqn:relaxed_optimization}
    \begin{eqnarray}
        &&\min_{\Delta, \Hat{K}} \sum_{i,j=1}^{M} c_{ij} \delta_{ij} + \left\Vert Y^F - \Hat{K} \begin{bmatrix}
        U^P \\ Y^P \\ U^F
    \end{bmatrix} \right\Vert _F^2 \\
    &&\textrm{s.t.} \ \delta_{ij} \in \{ 0,1 \} \quad \forall \{ i,j \} \in \mathcal{V} \\
    && \qquad \Delta = \{ \delta_{ij} \}_{i,j=1:M} \\
    && \qquad \Hat{K} = [\Hat{K}_p^u \; \Hat{K}_p^y \; \Hat{K}_f] \\
    && \qquad \Hat{K}_p^u = \{ \Hat{K}_{p,ij}^u \}_{i,j=1:M}, \  \Hat{K}_p^y = \{ \Hat{K}_{p,ij}^y \}_{i,j=1:M}, \  \Hat{K}_f = \{ \Hat{K}_{f,ij} \}_{i,j=1:M} \ \  \\
    \label{eqn:implies_constraint_bigM}
    && \! \! \! \! \left. \begin{array}{c}
    -\overline{M} \delta_{ij} \mathbbm{1}^{p_i N \times m_j T_{\textrm{ini}}} \leq \Hat{K}_{p,ij}^u \leq \overline{M} \delta_{ij} \mathbbm{1}^{p_i N \times m_j T_{\textrm{ini}}} \\
    -\overline{M} \delta_{ij} \mathbbm{1}^{p_i N \times p_j T_{\textrm{ini}}} \leq \Hat{K}_{p,ij}^f \leq \overline{M} \delta_{ij} \mathbbm{1}^{p_i N \times p_j T_{\textrm{ini}}} \\
    -\overline{M} \delta_{ij} \mathbbm{1}^{p_i N \times m_j N} \leq \Hat{K}_{f,ij} \leq \overline{M} \delta_{ij} \mathbbm{1}^{p_i N \times m_j N}
\end{array} \right\} \; \forall \{i,j\} \in \mathcal{V}, \; i \neq j,
    \end{eqnarray}
\end{subequations}
which may be represented as an MISOCP. Whilst this can take significant computation to solve for very large systems, we note that the communication topology is designed offline therefore computational efficiency is of reduced importance compared to an online algorithm.
%by reformulating the quadratic term in the objective as a conic constraint, vectorizing the decision variable $\Delta$ and forming a vector of communication costs. 
\par

In the case of noisy data, the solution to problem \eqref{eqn:relaxed_optimization} may give very low prediction error due to overfitting to the noise as we have dropped the rank constraint \eqref{eqn:rank_constraint}. A common approach to get around this issue is to collect $N_{\textrm{coll}} \in \mathbb{Z}_{>0}$ datasets using exactly the same input sequence, and then average the Hankel matrices obtained from the $N_{\textrm{coll}}$ experiments \cite{Huang-et-al-2019-DeePC-Grid-Connected-Power-Converters}. Since the noise is assumed to be zero-mean, the effect is to remove the noise from the data. The online measurements used to make predictions (e.g., in \eqref{eqn:spc_predictor}) will also be subject to noise; estimators such as the extended Kalman filter for output-feedback behavioral control \cite{Alpago-et-al-2020-EKF-DeePC} may be used to deal with this issue. Application of such approaches is outside the scope of this work as we focus on optimality of the communication topology and associated predictor and do not address online control. Supposing Assumption \ref{assumption:no_DFMs} holds, then a stabilizing control law exists for the optimized $K$. \par 

%Alternatively, regularizers are popular for providing robustness to noisy data \cite{Markovsky-Dorfler-2021-Behavioral-Systems-Theory-Analysis-Signal-Processing-Control}. The nuclear norm has been presented as a regularizer promoting low-rank matrices \cite{Markovsky-Dorfler-2021-Behavioral-Systems-Theory-Analysis-Signal-Processing-Control}. The Frobenius norm offers another possibility to essentially regularize the magnitude of the elements within the matrix. However, analysis showing that these regularizers offer a convex relaxation of the rank constraint within \eqref{eqn:original_optimization} are currently lacking. Numerical results considering these regularizers are presented below in Section \ref{section:simulation}.

\subsection{Bounding the Prediction Error}
\label{subsection:theory}

By considering the case where the communication graph is fully connected, we can deduce a bound on the prediction error when links are dropped. Intuitively, we expect that a graph with more links will lead to a lower prediction error; Theorem \ref{theorem} shows that the topology optimization scheme \eqref{eqn:relaxed_optimization} is consistent with this and that the prediction cost of \eqref{eqn:relaxed_optimization} under the optimized topology is bounded by the cost of the dropped links. We will use this to find an upper bound on the open-loop prediction error induced by dropping links. Adding a constraint that $\delta_{ij} = 1 \; \forall \{i,j\} \in \mathcal{V}, \ i \neq j$ leads to \eqref{eqn:relaxed_optimization} simplifying to an unconstrained optimization, assuming that $\overline{M}$ is sufficiently large:
\begin{eqnarray}
\label{eqn:opt_K_full_comms}
    \min_{\Hat{K}} J(\Hat{K}) = \min_{\Hat{K}} \sum_{i,j=1, \, i \neq j}^{M} c_{ij} + \left\Vert Y^F - \Hat{K} \begin{bmatrix}
        U^P \\ Y^P \\ U^F
    \end{bmatrix} \right\Vert _F^2.
\end{eqnarray}
We denote a minimizer of \eqref{eqn:opt_K_full_comms} by $K'$, and the value the cost function of \eqref{eqn:opt_K_full_comms} evaluated at the minimizer by $J'(K')$. The topology-predictor pair minimizing \eqref{eqn:relaxed_optimization} is given by $(\Delta^*,K^*)$, and the cost function evaluation at this minimizer is denoted as $J(\Delta^*,K^*)$, with $\Delta^* = \{\delta^*_{ij}\}_{i,j=1:M}$. 
\begin{theorem}
\label{theorem}
    Considering the $t$-th column of $Y^F$ as an $N$-step output trajectory $y_f^t$ of \eqref{eqn:global_dynamics} and the $t$-th column of $\Hat{Y}^F$ to be a prediction $\Hat{y}_f^t$ of $y_f^t$ where $1 \leq t \leq T - (T_{\textrm{ini}} + N)  + 1$, the sum of the Euclidean norm of the trajectory prediction errors $\sum_{t=1}^{T - (T_{\textrm{ini}} + N)  + 1}{\lVert y_f^t - \Hat{y}_f^t \rVert_2^2}$ induced by a given optimized topology $\Delta^*$ and predictor $K^*$ is lower bounded by the prediction error with a fully connected communication graph equal to $\left\Vert Y^F \big( I - \Pi \big) \right\Vert _F^2$, where $\Pi = \begin{bmatrix}
    U^P \\ Y^P \\ U^F
\end{bmatrix}^\dagger \begin{bmatrix}
        U^P \\ Y^P \\ U^F
    \end{bmatrix}$, and upper bounded by the sum of this lower bound with the sum of the communication costs of the links not included in the optimized topology, equal to $\left\Vert Y^F \big( I - \Pi \big) \right\Vert _F^2 + \sum_{i,j=1, \, i \neq j}^M{c_{ij} (1 - \delta_{ij}^*)}$.
\end{theorem}

\begin{proof}
    We write $J(\Delta^*,K^*) = \sum_{i,j=1}^{M} c_{ij} \delta^*_{ij} + \left\Vert Y^F - K^* \begin{bmatrix}
        U^P \\ Y^P \\ U^F
    \end{bmatrix} \right\Vert _F^2$.
%\begin{eqnarray}
%    J(\Delta^*,K^*) = \sum_{i,j=1}^{M} c_{ij} \delta^*_{ij} + \left\Vert Y^F - K^* \begin{bmatrix}
%        U^P \\ Y^P \\ U^F
%    \end{bmatrix} \right\Vert _F^2.
%\end{eqnarray}
Since $\Delta_1 = \{\delta_{ij}\}_{i,j=1:M}$ where $\delta_{ij} = 1 \; \forall \{i,j\} \in \mathcal{V}, \ i \neq j$ is a feasible topology for \eqref{eqn:relaxed_optimization}, which has a feasible optimal predictor $K'$ when $\overline{M}$ is sufficiently large, $(\Delta_1,K')$ is a feasible solution of \eqref{eqn:relaxed_optimization}. Therefore $J(\Delta^*,K^*) \leq J(\Delta_1,K') = J'(K')$, so $J(\Delta^*,K^*) - J'(K') \leq 0$:
\begin{eqnarray}
\label{eqn:pred_error_upper_bound_derive}
    \sum_{i,j=1}^{M} c_{ij} \delta^*_{ij} - \sum_{i,j=1, i \neq j}^{M} c_{ij} + \left\Vert Y^F - K^* \begin{bmatrix}
        U^P \\ Y^P \\ U^F
    \end{bmatrix} \right\Vert _F^2 - \left\Vert Y^F - K' \begin{bmatrix}
        U^P \\ Y^P \\ U^F
    \end{bmatrix} \right\Vert _F^2 \leq 0. \ 
\end{eqnarray}
Optimization \eqref{eqn:opt_K_full_comms} is unconstrained and the communication costs do not affect the location of the minimum, hence $K'$ is a global minimizer of $J'(\Hat{K})$ and the least-square prediction error given by $K' = Y^F \begin{bmatrix}
    U^P \\ Y^P \\ U^F
\end{bmatrix}^\dagger$. This naturally leads to:
%\vspace{-4mm}
\begin{eqnarray}
\label{eqn:lower_bound_pred_error}
    \left\Vert Y^F \Bigg(I -  \begin{bmatrix}
    U^P \\ Y^P \\ U^F
\end{bmatrix}^\dagger \begin{bmatrix}
        U^P \\ Y^P \\ U^F
    \end{bmatrix} \Bigg) \right\Vert _F^2 = \left\Vert Y^F \big( I - \Pi \big) \right\Vert _F^2 \leq \left\Vert Y^F - K^* \begin{bmatrix}
        U^P \\ Y^P \\ U^F
    \end{bmatrix} \right\Vert _F^2.
\end{eqnarray}
Considering the Frobenius norm as the sum of the 2-norm of the matrix columns (with each column a trajectory and a prediction respectively in the left and right matrix evaluations inside the Frobenius norm), from \eqref{eqn:pred_error_upper_bound_derive} we find:
\begin{eqnarray}
\label{eqn:pred_error_upper_bound}
    \sum_{t=1}^{T-(T_{\textrm{ini}} + N)+1}{\lVert y_f^t - \Hat{y}_f^t \rVert_2^2} \leq \sum_{i,j=1, \, i \neq j}^M{c_{ij} (1 - \delta_{ij}^*)} + \left\Vert Y^F \big( I - \Pi \big) \right\Vert _F^2.
\end{eqnarray}
Note we can always set $\delta_{ii}=0 \ \forall i \in \mathcal{V}$ without influencing prediction accuracy. Applying the same consideration to \eqref{eqn:lower_bound_pred_error} gives:
\begin{eqnarray}
    \left\Vert Y^F \big( I - \Pi \big) \right\Vert _F^2 \leq \sum_{t=1}^{T-(T_{\textrm{ini}} + N)+1}{\lVert y_f^t - \Hat{y}_f^t \rVert_2^2},
\end{eqnarray}
providing the upper and lower bounds on the total prediction error over all the samples in the collected dataset as desired. Note that the lower bound evaluates to zero when data is noise-free as exact predictions may be found.
    \qed
\end{proof}

Whilst bounds on the overall prediction error induced by a given topology are informative, it is useful to have bounds on the error for an individual trajectory prediction as would be used in solving a control problem. 
\begin{lemma}
\label{lemma}
    Consider a control output penalty term for a regulation problem of the form $\lVert y_f\rVert_Q^2$ where $Q \succeq 0$ is a weighting matrix. The error in the open-loop prediction cost, given by $\lVert y_f - \Hat{y}_f \rVert_Q^2$, is bounded by $\lambda_{Q,\textrm{max}} \big( \sum_{i,j=1, i \neq j}^M{c_{ij} (1 - \delta_{ij}^*)} + \left\Vert Y^F \big( I - \Pi \big) \right\Vert _F^2 \big)$, where $\lambda_{Q,\textrm{max}}$ is the largest eigenvalue of $Q$.
\end{lemma}
\begin{proof}
    A worst-case bound on any individual trajectory $y_f$ is found from \eqref{eqn:pred_error_upper_bound} as:
    \begin{eqnarray}
    \label{eqn:worst_case_pred_error_UB}
        \lVert y_f - \Hat{y}_f \rVert_2^2 \leq \sum_{i,j=1, \, i \neq j}^M{c_{ij} (1 - \delta_{ij}^*)} + \left\Vert Y^F \big(I - \Pi \big) \right\Vert _F^2.
    \end{eqnarray}
    We can write $Q = V_Q \Lambda_Q V_Q^\top$ since $Q$ is symmetric, where $V_Q$ is a matrix of eigenvectors of $Q$ and $\Lambda_Q$ is a matrix of its eigenvalues. Let $y_f - \Hat{y}_f = V_Q \Tilde{y}$, then $\lVert y_f - \Hat{y}_f \rVert_Q^2 = \Tilde{y}^\top V_Q^\top V_Q \Lambda_Q V_Q^\top V_Q \Tilde{y} = \Tilde{y}^\top \Lambda_Q \Tilde{y}$ since $V_Q^\top = V_Q^{-1}$. Note that $\lVert y_f - \Hat{y}_f \rVert_2^2 = \Tilde{y}^\top V_Q^\top V_Q \Tilde{y} = \Tilde{y}^\top \Tilde{y}$ and $\Tilde{y}^\top \Lambda_Q \Tilde{y} \leq \lambda_{Q,\textrm{max}} \Tilde{y}^\top \Tilde{y}$, therefore $\lVert y_f - \Hat{y}_f \rVert_Q^2 \leq \lambda_{Q,\textrm{max}} \lVert y_f - \Hat{y}_f \rVert_2^2$. Combining this fact with \eqref{eqn:worst_case_pred_error_UB} gives the result:
    \begin{eqnarray}
    \lVert y_f - \Hat{y}_f \rVert_Q^2 \leq \lambda_{Q,\textrm{max}} \Bigg( \sum_{i,j=1, \, i \neq j}^M{c_{ij} (1 - \delta_{ij}^*)} + \left\Vert Y^F \big( I - \Pi \big) \right\Vert _F^2 \Bigg).
\end{eqnarray}
    \qed
\end{proof}
This upper bound gives the worst-case output prediction weighted according to the control weight matrix $Q$; however, it is likely to be rather conservative. We next present simulation results using \eqref{eqn:relaxed_optimization} where we optimize communication topology and predictor and observe that the open-loop prediction cost error, bounded by Lemma \ref{lemma}, does indeed correlate with closed-loop control cost.

\section{Simulations}
\label{section:simulation}

We use the linearized and discretized swing dynamics as given in \cite{Alonso-2022-Data-Driven-Distributed-Localized-MPC} for topology optimization and control simulations. Each subsystem has dynamics \eqref{eqn:subsys_dynamics}, with $x_i(k) = [\theta_i(k) \; \omega_i(k)]^\top$, $B_i = \begin{bmatrix}
    0 & \frac{1}{m_i}
\end{bmatrix}^\top$, $C_i = \begin{bmatrix}
        1 & 0
    \end{bmatrix}$, $D_i = \mathbf{0}^{p_i \times m_i}$ and the matrices $A_i$ and $E_{ij}$ are given by:
\begin{equation}
\label{eqn:sim_system}
    A_i = \begin{bmatrix}
        1 & \Delta t \\
        -\frac{k_i}{m_i} \Delta t & 1 - \frac{d_i}{m_i} \Delta t
    \end{bmatrix} \quad E_{ij} = {\begin{bmatrix}
        0 \\
        \frac{k_{ij}}{m_i} \Delta t
    \end{bmatrix}},
\end{equation}
where $\theta_i$ and $\omega_i$ represent phase angle and frequency deviations, and $m_i$, $d_i$ and $k_{ij}$ represent inertia, damping and coupling for each subsystem, with the time step given by $\Delta t = 0.2$ and $k_i = \sum_{j=1/i}^M{k_{ij}}$. We choose $M=4$, $m_1 = 1.4$, $m_2 = 0.8$, $m_3 = 0.6$, $m_4 = 0.9$, $d_1 = 0.6$, $d_2 = d_4 = 0.65$, $d_3 = 0.75$, $k_{12} = k_{21} = 1.25$, $k_{23} = k_{32} = 1.2$, $k_{34} = k_{43} = 0.075$, $k_{13} = k_{14} = k_{24} = k_{31} = k_{41} = k_{42} = 0$.
%Parameters values used in the simulations are given in Table \ref{tab:sim_params}.
%\begin{table}[h]
%\caption{System parameter values used within simulations}
%\begin{center}
%\label{tab:sim_params}
%\begin{tabular}{r@{\quad}l@{\quad} r@{\quad} l}
%\hline
%\multicolumn{1}{c}{\rule{0pt}{12pt}
%                   Parameter \quad }&\multicolumn{1}{c}{Value \quad }&\multicolumn{1}{c}{Parameter  (\emph{cont.}) \quad } &\multicolumn{1}{c}{Value (\emph{cont})}\\[2pt]
%\hline\rule{0pt}{12pt}
%$m_1$  &     3 & $k_{12}$ & 0 \\
%$m_2$  &   0.5 & $k_{13}$ & 0.001 \\
%$m_3$  &   2 & $k_{14}$ & 0.05 \\
%$m_4$  & 1.5 & $k_{21}$ & 0.5 \\
%$d_1$  & 1 & $k_{23}$ & 0 \\
%$d_2$ & 0.75 & $k_{24}$ & 0 \\
%$d_3$ & 1 & $k_{31}$ & 0.5 \\
%$d_4$ & 1.5 & $k_{32}$ & 0 \\
%$\Delta t$ & 0.1 & $k_{34}$ & 0.002 \\
%& & $k_{41}$ & 0 \\
%& & $k_{42}$ & 0.0045 \\
%& & $k_{43}$ & 0.075 \\
%[2pt]
%\hline
%\end{tabular}
%\end{center}
%\end{table}
%\subsection{Communication Topology Optimization}
A persistently exciting input signal $u_{[1:T]}^{(g),d}$ was created using a pseudo-randomly generated sequence of zero-mean normally distributed numbers with variance $\sigma_d^2 = 1$ and injected into the test system within MATLAB simulations. The collected output data $y_{[1:T]}^{(g),d}$ and input $u_{[1:T]}^{(g),d}$ form the data Hankel matrices \eqref{eqn:Hankel_matrices} in order to solve \eqref{eqn:relaxed_optimization}. We choose $T_{\textrm{ini}} = 3$, $N = 5$, and $\overline{M} = 5$. Following tuning described below, $T$ and $N_{\textrm{coll}}$ were chosen as $200$ and $50$, respectively. \par
%hyperparameters chosen for the optimization are given in Table \ref{tab:sim_hyperparams}. 
The optimization \eqref{eqn:relaxed_optimization} is solved in MATLAB using GUROBI \cite{gurobi} with the YALMIP toolbox \cite{Lofberg-2004-YALMIP}, and we use an identical link cost $c_{ij} = c \; \forall \{i,j\} \in \mathcal{V}$. We test prediction accuracy and assess if overfitting to noise is occurring by injecting a random signal (each element sampled from a zero-mean normal distribution with variance $\sigma_d^2$) into each subsystem and calculating the mean-squared error (MSE) of predictions compared to true (denoised) outputs. Results for an SNR of $10^3$ averaged over 50 trials are presented in Table \ref{tab:sim_results}, giving the adjacency matrix for $\mathcal{G}_C$, optimization prediction cost and prediction MSE for different communication costs. Links are first dropped between subsystems that are uncoupled before links between coupled subsystems are dropped. Results show that half the links can be dropped before significant prediction errors occur; note that there are $6$ couplings (non-zero $k_{ij}$) out of a possible maximum of $12$. 
%\par
%\begin{table}[h]
%    \caption{Simulation hyperparameters}
%    \label{tab:sim_hyperparams}
%    \begin{center}
%    \begin{tabular}{c c c c c c c}
%    \hline
%    \multicolumn{1}{c}{\rule{0pt}{12pt}
%    $N$} & \multicolumn{1}{c}{$T_{\textrm{ini}}$} & \multicolumn{1}{c}{$T$} & \multicolumn{1}{c}{$N_{\textrm{coll}}$} & \multicolumn{1}{c}{$\textrm{SNR}$} & \multicolumn{1}{c}{$\sigma_d^2$} & \multicolumn{1}{c}{$\overline{M}$}\\[2pt]
%    \hline\rule{0pt}{12pt}
%        $10$ & $3$ & $200$ & $50$ & $10^3$ & $1$ & $5$  \\
%[2pt]
%\hline
%    \end{tabular}
%    \end{center}
%\end{table}
%\vspace{-10pt}
\begin{table}[ht]
\caption{Optimization results with an SNR of $10^3$, showing that links are dropped as the communication cost grows, resulting in increasing prediction cost and MSE}
\renewcommand{\arraystretch}{0.8}
\begin{center}
\label{tab:sim_results}
\begin{tabular}{c c c c c}
\hline
%& & $\textrm{SNR} = 10^3$ & & \\
\multicolumn{1}{c}{\rule{0pt}{12pt}
                   $c$ \quad }&\multicolumn{1}{c}{\# Links \quad }&\multicolumn{1}{c}{Topology \quad } &\multicolumn{1}{c}{Pred. Cost \quad} & \multicolumn{1}{c}{MSE}
                   %\multicolumn{1}{c}{\rule{0pt}{12pt}
                   %$c$ \quad }&\multicolumn{1}{c}{\# Links \quad }&\multicolumn{1}{c}{Topology \quad } &\multicolumn{1}{c}{Pred. Cost \quad} & \multicolumn{1}{c}{MSE $\times 10^{-3}$}
                   \\[2pt]
\hline\rule{0pt}{12pt}
$0.001$  &  12 & $\begin{bmatrix}
    0 & 1 & 1 & 1 \\
    1 & 0 & 1 & 1 \\
    1 & 1 & 0 & 1 \\
    1 & 1 & 1 & 0
\end{bmatrix}$ & $0.31$ & $0.100$ \\
\\
%$0.1$  &   10 & $\begin{bmatrix}
%    0 & 1 & 0 & 1 \\
%    1 & 0 & 1 & 1 \\
%    1 & 0 & 0 & 1 \\
%    1 & 1 & 1 & 0
%\end{bmatrix}$ & $2.73$ & $1.7$ \\
%\\
$1$  &   6 & $\begin{bmatrix}
    0 & 1 & 0 & 0 \\
    1 & 0 & 1 & 0 \\
    0 & 1 & 0 & 1 \\
    0 & 0 & 1 & 0
\end{bmatrix}$ & $0.68$ & $0.119$ \\
\\
$20$  & 4 & $\begin{bmatrix}
    0 & 1 & 0 & 0 \\
    1 & 0 & 1 & 0 \\
    0 & 1 & 0 & 0 \\
    0 & 0 & 0 & 0
\end{bmatrix}$ & $16.73$ & $0.457$ \\
\\
$1000$ & $2$ & $\begin{bmatrix}
    0 & 0 & 0 & 0 \\
    0 & 0 & 1 & 0 \\
    0 & 1 & 0 & 0 \\
    0 & 0 & 0 & 0
\end{bmatrix}$ & $1034.71$ & $1.247$ \\
[2pt]
\hline
\end{tabular}
\renewcommand{\arraystretch}{1}
\end{center}
\vspace{-6mm}
\end{table}
%\subsection{Hyperparameter Tuning}
%\label{subsection:hyperparam_tuning}
Our approach is to reduce the effects of noise by averaging over $N_{\textrm{coll}}$ datasets and using a $T > T_{\textrm{min}}$. We investigate varying these hyperparameters to provide sufficient performance without unduly large computational burdens. Results averaged over $500$ trials based on an all-to-all topology are shown in Figure \ref{fig:MSE_regularizer_tuning}. MSE converges as $T$ and $N_{\textrm{coll}}$ increase; we choose $T=200$ and $N_{\textrm{coll}}=50$ to benefit both topology optimization and data-driven control algorithms. 
\begin{figure}[ht]
    \centering
    \begin{subfigure}[t]{0.49\textwidth}
        \includegraphics[width=\textwidth, height=0.575\textwidth]{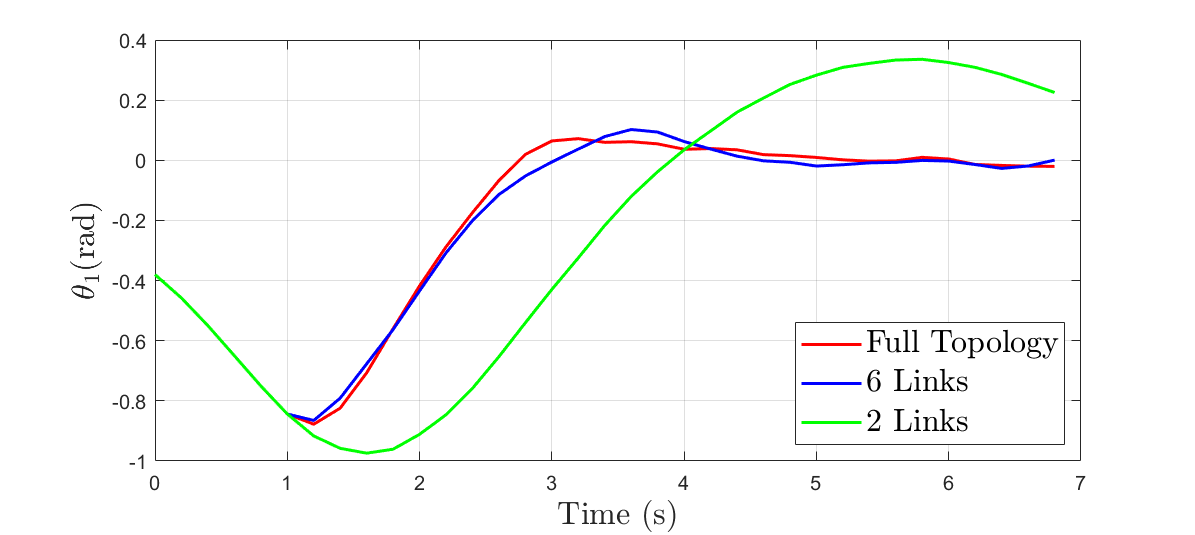}
        \caption{Evolution of $\theta_1$ for all-to-all communication and optimized topologies with $2$ and $6$ links}
        \label{fig:theta4_evolution}
    \end{subfigure}
    \hfill
    \begin{subfigure}[t]{0.49\textwidth}
    \includegraphics[width=\textwidth, height=0.55\textwidth]{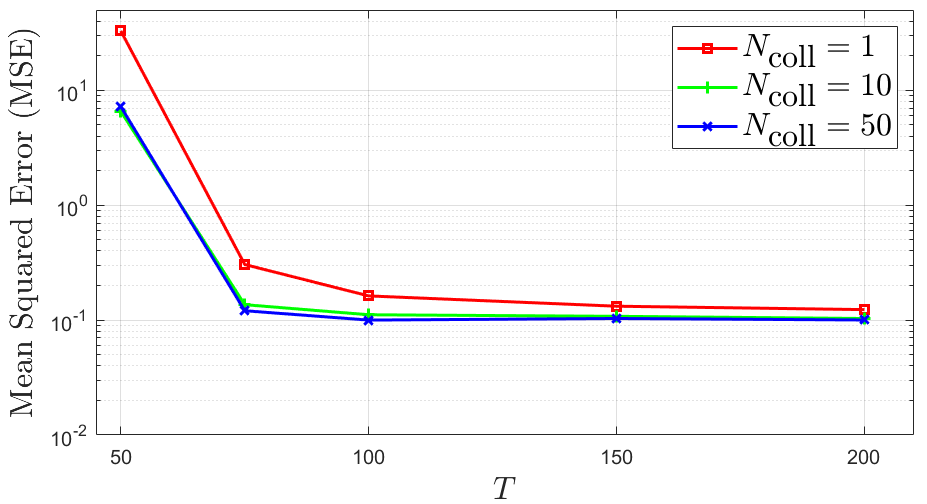}
    \caption{MSE decreases with increasing $T$ and $N_{\textrm{coll}}$ (results for an SNR of $10^3$)}
    \label{fig:MSE_regularizer_tuning}
    \end{subfigure}
    \caption{Simulation results for (a) no noise and (b) tuning $T$ and $N_{\textrm{coll}}$}
\end{figure}
\begin{figure}[ht]
    \centering
    \begin{subfigure}[t]{0.49\textwidth}
        \includegraphics[width=\textwidth, height=0.55\textwidth]{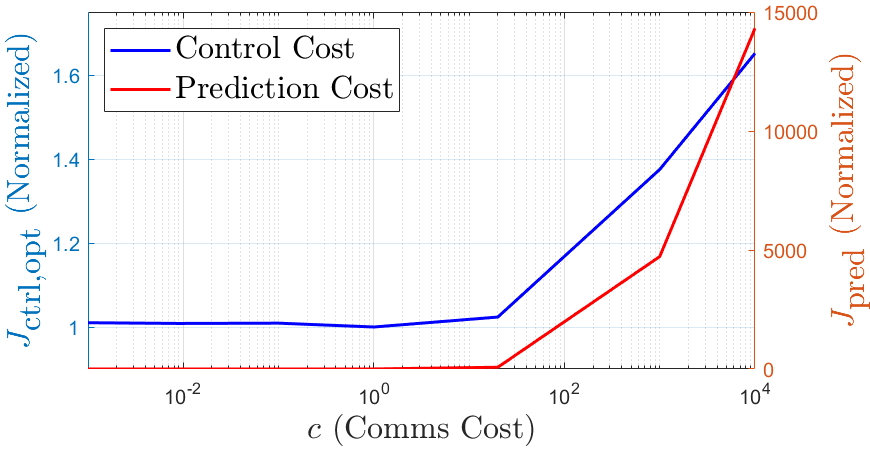}
        \caption{Control cost and prediction cost from topology optimization are closely correlated}
        \label{fig:optimized_ctrl_cost}
    \end{subfigure}
    \hfill
    \begin{subfigure}[t]{0.49\textwidth}
    \includegraphics[width=\textwidth, height=0.55\textwidth]{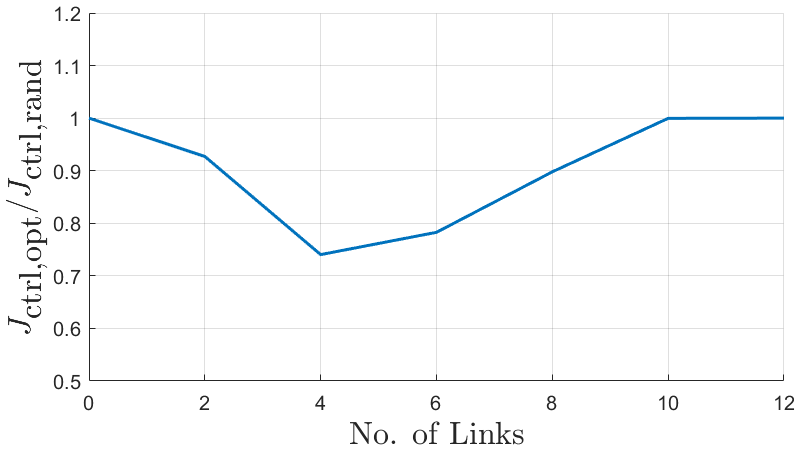}
    \caption{Ratio of control cost for optimized topologies and a random topology with the same number of links}
    \label{fig:control_cost_ratio}
    \end{subfigure}
    \caption{Control simulation results showing (a) control cost $J_{\textrm{ctrl,opt}}$ and prediction cost $J_{\textrm{pred}}$ for optimized topologies and (b) the ratio of $J_{\textrm{ctrl,opt}}$ to random topology control cost $J_{\textrm{ctrl,rand}}$}
    \vspace{-4mm}
\end{figure}

Finally, we investigate the influence of topology choice on control cost and its correlation with prediction cost using test system \eqref{eqn:sim_system}. 
%We use the same test system \eqref{eqn:sim_system} but modify the parameters to $m_1 = 0.8$, $m_2 = 1$, $m_3 = 0.6$, $m_4 = 1$, $d_1 = d_2 = d_3 = d_4 = 0.2$, $k_{13} = k_{23} = k_{24} = 0$, $k_{12} = 1.25$, $k_{14} = 2.5$, $k_{21} = 3$, $k_{31} = 1$, $k_{32} = 0.1$, $k_{34} = 7.5$, $k_{41} = 0.75$, $k_{42} = 0.5$, and $k_{43} = 2$ in order to increase the difference in control performance between all-to-all and decentralized topologies. 
Control is implemented using a non-cooperative data-driven MPC scheme based on \cite{Kohler-2022-Data-Driven-Distributed-MPC-Coupled-Linear} with a minor extension using slack variables, similarly to \cite{Coulson-2019-Shallows-of-the-DeePC}, to guarantee feasibility and stability with measurement noise. Output trajectories for $\theta_1$ shown in Figure \ref{fig:theta4_evolution} for $3$ optimized topologies.
%$k_{32} = k_{41} = 0$, $k_{13} = 0.001$, $k_{14} = 0.05$, $k_{21} = k_{31} = 0.5$, $k_{34} = 0.002$, $k_{42} = 0.0045$, and $k_{43} = 0.075$.
There is a clear correlation between prediction cost of the topology optimization and closed-loop control cost, shown in Figure \ref{fig:optimized_ctrl_cost}, in accordance with our approach to use prediction cost as a metric for topology optimality. Note the significant difference in growth rate of control and prediction costs. The control cost for an optimized topology consistently outperforms a random topology, shown in Figure \ref{fig:control_cost_ratio}; topologies with both $0$ and $12$ links cannot be random so these costs are equal. There is a maximum reduction of around $25$ \% in control cost for the optimized topology compared to the average random topology. This demonstrates the capability of the scheme to find an effective communication topology for distributed control.
%The outlier at $5$ links may be caused by testing over an insufficient number of random topologies, whilst the underperformance with $1$ link may be due to an optimal predictor not directly corresponding to the optimal choice for closed-loop control.

\vspace{-1.5mm}
\section{Conclusions and Outlook}
\label{section:conclusion}
\vspace{-1mm}

We have proposed a fully data-driven method for optimizing the communication topology of distributed control schemes via an MISOCP for LTI cyber-physical systems that do not have any unstable DFMs, assuming knowledge of a persistently exciting input sequence of sufficient order. The prediction cost under the optimized topology is upper bounded by the sum of the link costs not included in the topology with the prediction cost under all-to-all communication. Simulations demonstrated that the optimized topology achieves low prediction error up to a threshold where prediction accuracy decreases as further links are dropped. 
%and the benefits of using a longer collected dataset trajectory and averaging over multiple datasets to reduce the effects of measurement noise. 
Control simulations showed that the prediction cost of the topology optimization is correlated with realized control cost, and demonstrated reductions in control cost using the optimized topology compared to a random topology. Future work aims to improve scalability via distributed optimization, further analyse the connection between prediction cost and closed-loop control cost, and investigate data-driven methods to guarantee stability when unstable DFMs exist.

%\begin{proposition}

%\end{proposition}
%
%\begin{proof}
% \qed
%\end{proof}
%
%\begin{corollary}

%$\xi_{1},\allowbreak\dots,\allowbreak\xi_{N}$  be the
%\end{corollary}
%

%$\alpha \in \bbbr$, we mean the $a\in \bbbz$

%

%\begin{lemma}
%\end{lemma}
%
%
%\paragraph{Notes and Comments.}

%\begin{table}
%\caption{This is the example table taken out of {\it The
%\TeX{}book,} p.\,246}
%\begin{center}
%\begin{tabular}{r@{\quad}rl}
%\hline
%\multicolumn{1}{l}{\rule{0pt}{12pt}
%                   Year}&\multicolumn{2}{l}{World population}\\[2pt]
%\hline\rule{0pt}{12pt}
%8000 B.C.  &     5,000,000& \\
%  50 A.D.  &   200,000,000& \\
%1650 A.D.  &   500,000,000& \\
%1945 A.D.  & 2,300,000,000& \\
%1980 A.D.  & 4,400,000,000& \\[2pt]
%\hline
%\end{tabular}
%\end{center}
%end{table}
%
%\begin{theorem} [Ghoussoub-Preiss]\label{ghou:pre}
%\qed
%\end{theorem}
%
%\begin{example} [{{\rm External forcing}}]
%\end{example}
%
%\begin{definition}

%\end{definition}
%


\bibliographystyle{splncs03}
\footnotesize{
\bibliography{bibliography.bib}}
%
%\begin{thebibliography}{6}
%

%\bibitem {smit:wat}
%Smith, T.F., Waterman, M.S.: Identification of common molecular subsequences.
%J. Mol. Biol. 147, 195?197 (1981). \url{doi:10.1016/0022-2836(81)90087-5}

%\bibitem {may:ehr:stein}
%May, P., Ehrlich, H.-C., Steinke, T.: ZIB structure prediction pipeline:
%composing a complex biological workflow through web services.
%In: Nagel, W.E., Walter, W.V., Lehner, W. (eds.) Euro-Par 2006.
%LNCS, vol. 4128, pp. 1148?1158. Springer, Heidelberg (2006).
%\url{doi:10.1007/11823285_121}

%\bibitem {fost:kes}
%Foster, I., Kesselman, C.: The Grid: Blueprint for a New Computing Infrastructure.
%Morgan Kaufmann, San Francisco (1999)

%\bibitem {czaj:fitz}
%Czajkowski, K., Fitzgerald, S., Foster, I., Kesselman, C.: Grid information services
%for distributed resource sharing. In: 10th IEEE International Symposium
%on High Performance Distributed Computing, pp. 181?184. IEEE Press, New York (2001).
%\url{doi: 10.1109/HPDC.2001.945188}

%\bibitem {fo:kes:nic:tue}
%Foster, I., Kesselman, C., Nick, J., Tuecke, S.: The physiology of the grid: an open grid services architecture for distributed systems integration. Technical report, Global Grid
%Forum (2002)

%\bibitem {onlyurl}
%National Center for Biotechnology Information. \url{http://www.ncbi.nlm.nih.gov}

%\end{thebibliography}
\end{document}